\documentclass[11pt]{article}
\usepackage[margin=1in]{geometry}
\usepackage{graphicx}
\usepackage{xspace}
\usepackage{xcolor}
\usepackage{amscd}
\usepackage{amsmath}
\usepackage{amssymb}
\usepackage{amstext}
\usepackage{amsthm}
\usepackage{bbold}
\usepackage{bm}
\usepackage{colonequals}
\usepackage{mathtools}
\usepackage{tikz}
\usepackage{booktabs,tabularx,array}
\usepackage{enumitem}
\usepackage{float}
\usepackage{algpseudocode}
\usepackage[ruled,vlined]{algorithm2e}
\usepackage{microtype}
\usepackage{hyperref}

\newcommand{\sE}{\mathcal{E}}

\newcommand{\sL}{\mathcal{L}}

\newcommand{\FF}{\mathbb{F}}

\newcommand{\PP}{\mathbb{P}}
\newcommand{\QQ}{\mathbb{Q}}

\DeclareSymbolFont{bbold}{U}{bbold}{m}{n}
\DeclareSymbolFontAlphabet{\mathbbold}{bbold}

\newcommand{\Adv}{\mathop{\mathsf{Adv}}}
\newcommand{\Var}{\mathrm{Var}}

\DeclareMathOperator*{\E}{\mathbb{E}}

\newcommand{\Unif}{\mathrm{Unif}}

\newcommand{\Bin}{\mathrm{Bin}}
\newcommand{\eps}{\varepsilon}

\newcommand{\RM}{\operatorname{RM}}
\newcommand{\Alt}{\operatorname{Alt}}
\newcommand{\rank}{\operatorname{rank}}
\newcommand{\bias}{\operatorname{bias}}
\newcommand{\Tnoise}{\mathsf{T}}

\newtheorem{theorem}{Theorem}[section]
\newtheorem{lemma}[theorem]{Lemma}
\newtheorem{proposition}[theorem]{Proposition}
\newtheorem{corollary}[theorem]{Corollary}
\newtheorem{conjecture}[theorem]{Conjecture}

\theoremstyle{definition}

\theoremstyle{remark}
\newtheorem{remark}[theorem]{Remark}

\setlist[itemize]{leftmargin=1.5em,itemsep=0.3em,topsep=0.4em}
\setlist[enumerate]{leftmargin=1.7em,itemsep=0.3em,topsep=0.4em}
\newcolumntype{Y}{>{\raggedright\arraybackslash}X}

\allowdisplaybreaks

\title{The Polynomial-Time Low-Degree Conjecture is False}
\author{Songtao Mao\thanks{\texttt{smao13@jhu.edu}, Department of Computer Science, Johns Hopkins University.}}
\date{}

\begin{document}
\maketitle

\begin{abstract}
The low-degree method and its associated lower bounds are widely used to guide algorithm design and to provide evidence of computational hardness in average-case inference, high-dimensional statistics, random optimization, and related problems. This led to the low-degree conjecture, which predicts that when the low-degree advantage between a planted distribution and a uniform null distribution remains bounded, no efficient distinguisher can succeed after independent noise, provided that the planted distribution has permutation symmetry. Several works have produced counterexamples to variants of this conjecture or to versions for algorithms with higher time complexity, but the conjecture remained open in its standard binary, polynomial-time formulation.

We disprove the polynomial-time low-degree conjecture by giving a family of examples. For every fixed integer $r\geq3$, we construct a permutation-invariant distribution $\mathbb{P}_n$ on simple graphs, with $\mathbb{Q}_n=G(n,1/2)$, such that every marginal of $\mathbb{P}_n$ on at most $D_n=\Theta((\log n)^{r-1})$ edges is uniform. Therefore, the low-degree advantage is zero through degree $D_n$. Nevertheless, after every edge is independently resampled at a fixed positive rate, a deterministic rank test strongly distinguishes the resulting distribution from $\mathbb{Q}_n$ in linear time in the graph input size.

The construction chooses a subspace of a Reed--Muller code whose nonzero polynomials have small absolute bias, selects points whose evaluation vectors have no short linear dependencies, and evaluates a random alternating bilinear form on pairs of these vectors. Our result shows that low-degree indistinguishability, a uniform null distribution, permutation invariance, and independent resampling do not by themselves imply polynomial-time hardness, and suggests that a valid general conjecture must impose an additional condition.
\end{abstract}
\newpage
\tableofcontents
\newpage

\section{Introduction}
\label{sec:introduction}

Many high-dimensional inference problems exhibit a statistical--computational gap: a signal is detectable statistically, yet no efficient detection algorithm is known. Understanding such gaps is a central problem in average-case complexity. A common approach is to study the limitations of specific algorithms or algorithmic frameworks, a perspective that motivates the low-degree method.

\paragraph{The low-degree method.}

The low-degree method studies bounded-degree polynomials as a model of efficient computation and uses their performance to predict statistical--computational gaps. It emerged from the study of the sum-of-squares hierarchy, where pseudo-calibration constructions linked semidefinite relaxations to low-degree polynomial tests \cite{BHKMP2019}. Building on this connection, Hopkins et al.\ \cite{HKPSS2017,HopkinsSteurer2017,Hopkins2018} proposed low-degree indistinguishability as evidence of computational hardness. Kunisky, Wein, and Bandeira \cite{KWB2019} later developed the likelihood-ratio formulation now commonly used for hypothesis testing.

Given one observation, we ask whether it comes from a null distribution or a planted distribution. We use standard hypothesis-testing terminology to describe this problem. For a distribution $\mathsf R$, write $\E_{\mathsf R}$ and $\Var_{\mathsf R}$ for expectation and variance. Let $\PP$ and $\QQ$ be distributions on $\{0,1\}^N$ with $\PP\ll\QQ$, and let $\sL=\mathrm d\PP/\mathrm d\QQ$. Write $\sL^{\leq d}$ for the orthogonal projection of $\sL$ in $L^2(\QQ)$ onto the real multilinear polynomials of degree at most $d$, where $\|f\|_{L^2(\QQ)}=(\E_{\QQ}[f^2])^{1/2}$. For $d\geq0$, define
\begin{equation}
  \label{eq:lda}
  \Adv_{\leq d}(\PP,\QQ)
  :=\sup_{\substack{f:\,\deg(f)\leq d\\\Var_{\QQ}(f)>0}}
  \frac{|\E_{\PP}f-\E_{\QQ}f|}{\sqrt{\Var_{\QQ}(f)}}
  =\left\|\sL^{\leq d}-1\right\|_{L^2(\QQ)}.
\end{equation}
The advantage is zero exactly when $\PP$ and $\QQ$ agree on the expectations of all real polynomials of degree at most $d$; for a uniform null, this is equivalent to uniformity of every marginal of $\PP$ on at most $d$ coordinates.

The low-degree likelihood ratio is often used as a statistic to guide algorithm design or provide evidence of limitations for specific classes of algorithms. Graph applications include community detection, planted dense subhypergraphs, and correlated or geometric random graphs \cite{HopkinsSteurer2017,ChenDingGongLi2026,DhawanMaoWein2025,DingDuLi2025,BangachevBresler2024}. Related work studies local-statistics semidefinite relaxations for graph certification \cite{BBKMW2021} and graph lifts \cite{KuniskyYu2024}. Further applications include matrix and tensor estimation and certification \cite{MaoWein2025,BandeiraKuniskyWein2020,DKWB2024,Li2025CorrelatedSpiked,ChooDorsi2021,Wein2023TensorDecomposition,DKWB2021}, graphon estimation \cite{SchrammWein2022,LuoGao2024}, random constraint satisfaction and optimization \cite{BreslerHuang2022,Wein2021MIS,GamarnikJagannathWein2024,Mao2026}, and cryptography \cite{BogdanovKothariRosen2023}. The Franz--Parisi criterion connects low-degree bounds to free-energy methods \cite{BandeiraEtAl2022FranzParisi}. Under suitable assumptions, low-degree tests have rigorous connections to statistical-query algorithms \cite{BBHLS2021}; in Gaussian rank-one matrix estimation, constant-degree polynomial estimators are connected to approximate message passing \cite{MontanariWein2025}. See \cite{KWB2019,Wein2025Survey} for surveys.

\paragraph{The low-degree conjecture.}

Do the limitations of low-degree polynomials also capture the limitations of all efficient algorithms? In many inference problems, low-degree bounds match the performance of the best known algorithms, suggesting that this may hold under suitable assumptions.

Hopkins \cite{Hopkins2018} formulated this correspondence as a general conjecture. In the polynomial-time form considered here, if polynomials of degree at most $(\log n)^{1+c}$ have only bounded advantage for some constant $c>0$, then any fixed positive rate of independent resampling should make polynomial-time strong distinction impossible.

Under suitable versions of the low-degree conjecture, low-degree bounds imply limits on efficient testing error rates and computational barriers for community recovery, graph matching, and mixed sparse regression \cite{MoitraWein2025,DingHuaSlotSteurer2025,Li2025Contiguity,ArpinoVenkataramanan2023}.

We now present a binary formulation, one of the most common forms of the low-degree conjecture. For graph-valued observations, let $[n]=\{1,\ldots,n\}$ and $\sE_n=\binom{[n]}2$. For a permutation $\tau$ of $[n]$ and $x\in\{0,1\}^{\sE_n}$, define $(\tau x)_{\{i,j\}}:=x_{\{\tau^{-1}(i),\tau^{-1}(j)\}}$; a graph distribution is permutation-invariant if $\tau X$ and $X$ have the same distribution for every $\tau$. If $\QQ_n$ is the uniform distribution on $\{0,1\}^{\sE_n}$ and $\mathsf R_n$ is another distribution on the same space, define $\Tnoise_\eps\mathsf R_n$ by independently replacing each edge of a sample from $\mathsf R_n$, with probability $\eps\in[0,1]$, by a fresh uniform bit. A possibly randomized test $\varphi_n$ strongly distinguishes $\mathsf R_n$ from $\QQ_n$ if
\[
  \Pr_{Y\sim\mathsf R_n}\{\varphi_n(Y)=\textsc{Planted}\}=1-o(1)
  \quad\text{and}\quad
  \Pr_{Y\sim\QQ_n}\{\varphi_n(Y)=\textsc{Null}\}=1-o(1),
\]
where the probabilities include the internal randomness of the test.

The following is the graph-valued, polynomial-time form of Hopkins's conjecture \cite[Conjecture~2.2.4]{Hopkins2018}, expressed in terms of low-degree advantage.

\begin{conjecture}[Polynomial-time low-degree conjecture]\label{conj}
Let $\QQ_n$ be the uniform distribution on $\{0,1\}^{\binom n2}$, and let $\PP_n$ be a distribution on the same space that is invariant under every relabeling of the $n$ vertices. Suppose that there are a constant $c>0$ and integers $D_n\geq(\log n)^{1+c}$ such that
\[
  \Adv_{\leq D_n}(\PP_n,\QQ_n)=O(1).
\]
Then, for every fixed resampling rate $\eps\in(0,1)$, no polynomial-time test strongly distinguishes $\Tnoise_\eps\PP_n$ from $\QQ_n$.
\end{conjecture}

The conjecture has specific distributional assumptions; we explain them using prior work:
\begin{itemize}[itemsep=0.5em]
\item \textbf{Constant noise rate.}
The noise requirement rules out tests that rely on fragile exact algebraic relations. Gaussian elimination and lattice-based algorithms \cite{DK2022,ZSWB2022} can surpass low-degree predictions in noiseless settings. The real-valued counterexample in \cite{HW2021} exploits unresampled coordinates, motivating Ornstein--Uhlenbeck noise for Gaussian models.

\item \textbf{Symmetry of the planted distribution.}
Vertex-relabeling invariance means that the planted law has no preferred vertex labeling. Without this condition, efficient decoding of a suitable binary linear code gives a direct counterexample \cite{HW2021}. Large dual distance makes a uniformly random codeword exactly uniform on small sets of coordinates, while sufficient distance and efficient decoding distinguish a slightly corrupted codeword from a uniform string, which is typically far from the code. Under the stronger symmetry of invariance under all coordinate permutations, Hsieh et al. \cite{hsieh2026rigorous} show that vanishing low-degree advantage against the uniform null through sufficiently high degree implies statistical indistinguishability after fixed noise.

\item \textbf{Uniform null model.}
For real-valued models with nonproduct nulls, polynomial-time distinction is possible using the largest eigenvalue in the matrix model of \cite{BHJK2025}, and also for robust subspace recovery \cite{JiaVijayaraghavan2026}. Both tolerate noise in their respective models.
\end{itemize}

Related formulations use the low-degree likelihood ratio and consider algorithms with higher running times \cite{KWB2019,DKWB2021,DKWB2024,BHJK2025}. In particular, the fine-grained heuristic associates higher degree bounds with larger running times, as proposed by Hopkins~\cite{Hopkins2018} and discussed in~\cite{KWB2019}. Under the conjecture's assumptions, bounded advantage through degree $D$ is taken as evidence that strong distinction is hard at the running-time scale $n^{\widetilde{\Theta}(D)}$, with the tilde allowing polylogarithmic factors in the exponent. This prediction is refuted in \cite{BHJK2025}: the symmetric Boolean graph examples have zero advantage through degree $D=n^{\Omega(1)}$ but admit an $n^{O(\log n)}$-time distinguisher after sufficiently small fixed positive-rate resampling, leaving open the polynomial-time question under the conjecture's assumptions. This polynomial-time question is central to using low-degree bounds as evidence of computational hardness. We resolve it by giving a distribution that has zero low-degree advantage through superlogarithmic degree but admits strong distinction from the uniform null after sufficiently small fixed-rate resampling in $O(n^2)$ time, linear in the input size.

\subsection{Our result}

We state our main result, which disproves Conjecture~\ref{conj}.

\begin{theorem}[Conjecture~\ref{conj} is false]
\label{thm:main}
For every fixed integer $r\geq3$, there are constants $c_r>0$ and $\eps_r\in(0,1)$, and a sequence of permutation-invariant distributions $\PP_n$ on $\{0,1\}^{\sE_n}$ such that, for
\[
  D_n=\left\lfloor c_r(\log_2 n)^{r-1}\right\rfloor,
\]
the following statements hold.
\begin{enumerate}
  \item For all sufficiently large $n$, every marginal of $\PP_n$ on at most $D_n$ edges is exactly uniform; in particular, $\Adv_{\leq D_n}(\PP_n,\QQ_n)=0$.
  \item A single deterministic algorithm running in time $O(n^2)$ strongly distinguishes $\Tnoise_\eps\PP_n$ from $\QQ_n$ for every fixed $\eps\in[0,\eps_r]$.
\end{enumerate}
Moreover, there is a randomized nonuniform algorithm that uses $O(n\log n+(\log n)^{2r})$ bits of advice, samples exactly from $\PP_n$, and runs in expected polynomial time.
\end{theorem}

A sequence of distributions $(\mathsf R_n)$ is uniformly samplable if a single randomized algorithm, given $1^n$, samples exactly from $\mathsf R_n$ in expected time polynomial in $n$. Our distinguisher is uniform and requires no knowledge of the advice; the advice is used only by the sampler.

Compared with \cite{BHJK2025}, our construction fools polynomials only up to a smaller degree, but the graph distinguisher runs in polynomial rather than quasipolynomial time.

The same planted family also admits an $o(n)$-time strong distinguisher with random access to the input after sufficiently small fixed-rate resampling; see Section~\ref{sec:discussion}.

\paragraph{Extensions.}
The construction extends to symmetric tensors of any fixed order $k\geq2$ and, using the embedding in \cite{KlivansKothari2014} and the noise analysis in \cite{BHJK2025}, to observations with a product standard Gaussian null. Both extensions preserve vertex-relabeling invariance and exact low-degree agreement, and admit polynomial-time distinction after sufficiently small fixed noise. Section~\ref{sec:discussion} describes the tensor extension under resampling noise and the Gaussian extension under Ornstein--Uhlenbeck noise.

\subsection{Technical overview}

The distinguisher in \cite{BHJK2025} uses Reed--Solomon list decoding to distinguish the noisy planted distribution from the null. A random vertex permutation hides the positions and order of $\Theta(\log n)$ rows that encode the evaluation points and their polynomial values. The algorithm searches over these ordered row indices and runs a polynomial-time list decoder for each choice, giving an $n^{O(\log n)}$ running time.

To avoid this search, an efficiently computable statistic beyond low-degree moments is needed. Noise can destroy the linear dependencies detected by Gaussian elimination on the original adjacency matrix. Our distinguishing algorithm therefore first applies a monomial lift to a small adjacency submatrix, then computes its rank by Gaussian elimination over $\FF_2$. For a submatrix $Y$ between $T$ anchor vertices and $N_0$ other vertices, the degree-$s$ lift $M_s(Y)$ has the form
\[
  \underbrace{\begin{bmatrix}
    y_{11}&\cdots&y_{1N_0}\\
    \vdots&\ddots&\vdots\\
    y_{T1}&\cdots&y_{TN_0}
  \end{bmatrix}}_{T\times N_0}
  \xrightarrow{\text{lift each column}}
  \underbrace{\begin{bmatrix}
    1&\cdots&1\\
    y_{11}&\cdots&y_{1N_0}\\
    \vdots&&\vdots\\
    y_{11}y_{21}&\cdots&y_{1N_0}y_{2N_0}\\
    \vdots&&\vdots
  \end{bmatrix}}_{\binom{T}{\leq s}\times N_0}.
\]
The test declares ``planted'' exactly when $\rank M_s(Y)<N_0$.

The construction must give uniform small edge marginals while preserving a polynomial structure that the lift can expose. We use the Reed--Muller space $\RM(m,r)$ of degree-$r$ functions on $\FF_2^m$. It obeys a simple product rule: the pointwise product of at most $s$ such functions belongs to $\RM(m,rs)$.

\paragraph{Choosing the Reed--Muller subspace.}
We first choose a $K$-dimensional subspace $W\leq\RM(m,r)$ of the degree-$r$ Reed--Muller space over $\FF_2$, where $m=\Theta(\log n)$ and $K=\Theta((\log n)^r)$, such that every nonzero element of $W$ has small absolute bias. A result of Ben-Eliezer, Hod, and Lovett \cite{BHL2012} bounds the probability that a random Reed--Muller polynomial has large absolute bias; an averaging argument then produces the required subspace. Small bias is used to find the evaluation points. For a basis $f_1,\ldots,f_K$ of $W$ and a uniform point $x$, every nonzero parity of $(f_1(x),\ldots,f_K(x))$ is nearly unbiased. A Fourier calculation then bounds the probability that a short sum of independently sampled evaluation vectors is zero.

\paragraph{Choosing the evaluation points and defining the planted distribution.}
A union bound lets us choose $n$ evaluation points $x_1,\ldots,x_n$ whose vectors have no short linear dependencies. We fix these points, sample a uniformly random alternating bilinear form $B$, and define the edge indicators before relabeling by
\[
  X_{\{i,j\}}=B(v_i,v_j),\qquad
  v_i=(f_1(x_i),\ldots,f_K(x_i)),\qquad i<j.
\]
A final uniform vertex permutation makes the distribution invariant under relabeling.

It is easy to see that, for any small set of edges, the evaluation vectors at their endpoints are linearly independent, so after a change of basis the edge values are distinct coordinates of $B$ and hence independent uniform bits.

\paragraph{Clean columns after resampling.}
Under $\QQ_n$, the columns $y_j$ are independent and uniform in $\FF_2^T$, and a theorem of Bhandari et al. \cite{BHSS2022} on random puncturings of Reed--Muller codes gives full column rank after lifting with high probability.

Under the planted law, fixing $B$ and the anchors makes each noiseless edge a polynomial of degree at most $r$ in the other vertex's latent point. The product rule therefore keeps the lifted noiseless columns in a space of dimension at most $\binom{m}{\leq rs}$.

We next consider constant noise. For a sufficiently small fixed resampling rate, enough columns remain untouched before lifting to preserve a linear dependence with high probability. The remaining columns cannot remove this dependence, so the lifted matrix stays rank-deficient and the test still distinguishes the distributions. The rank bound on the clean columns holds for every assignment of latent points to these labels, so the test does not need to recover the hidden permutation or locate special rows.

\subsection{Organization and notation}
\paragraph{Organization.}
The rest of this paper is organized as follows. Section~\ref{sec:distributions} constructs $\PP_n$, proves that every small collection of edges is uniformly distributed, and derives the vanishing low-degree advantage. Section~\ref{sec:distinguisher} defines the rank test, analyzes it under both distributions, and proves that it runs in $O(n^2)$ time. Section~\ref{sec:discussion} explains the scope of the result and discusses some open problems.

\paragraph{Notation and conventions.}
For a finite set $S$, $\binom{S}{k}$ denotes the collection of its $k$-element subsets, while for an integer $M$, $\binom{M}{k}$ denotes their number. $\Unif(S)$ denotes its uniform distribution, and $X\sim\mathsf R$ means that $X$ has distribution $\mathsf R$. We write $G(n,p)$ for the distribution on graphs in which the edges are independently present with probability $p$. A variable subscript on $\Pr$ or $\E$, as in $\Pr_f$ or $\E_x$, indicates probability or expectation over that variable with the distribution specified in context; in particular, $\E_{x\in S}$ denotes uniform averaging over a finite set $S$. Unsubscripted $\Pr$ and $\E$ include all random choices then in force. We write $\Bin(N,p)$ for the binomial distribution with $N$ trials and success probability $p$.

Throughout, $r\geq3$ is fixed and $n\to\infty$; constants may depend on $r$, and all logarithms are base $2$. For positive sequences, we write $A_n\gg B_n$ when $A_n/B_n\to\infty$. The symbol $\FF_2$ denotes the binary field, $W\leq V$ means that $W$ is a linear subspace of $V$, and $a\cdot v$ is the standard dot product over $\FF_2$. The Reed--Muller space $\RM(m,d)$ consists of the functions $\FF_2^m\to\FF_2$ represented by multilinear polynomials of total degree at most $d$. Unless stated otherwise, the finite-dimensional algebraic objects in Sections~\ref{sec:distributions} and~\ref{sec:distinguisher}---vector spaces, bilinear forms, evaluation matrices, and matrix ranks---are over $\FF_2$. The spaces $L^2(\QQ_n)$ and their polynomial subspaces are real Hilbert spaces. Reed--Muller polynomials have coefficients in $\FF_2$, whereas the polynomials used as low-degree tests are real-valued functions of the graph's edge indicators. We write $\binom{m}{\leq d}=\sum_{j=0}^d\binom mj$ and $h_2(u)=-u\log_2u-(1-u)\log_2(1-u)$ for $u\in[0,1]$, with $0\log_20=0$. Graph edges are unordered, and degree means multilinear total degree. An event holds with high probability if it has probability $1-o(1)$.

\section{Distributions with zero low-degree advantage}
\label{sec:distributions}
In this section, we present the construction of $\PP_n$ and prove that its low-degree advantage against $\QQ_n$ is zero through degree $\Theta((\log n)^{r-1})$. We first choose the Reed--Muller subspace and evaluation points, and then use a random alternating form to define the graph distribution.

\subsection{Background on Reed--Muller codes}
The construction and the distinguisher use two results about Reed--Muller codes. Ben-Eliezer, Hod, and Lovett \cite{BHL2012} bound the probability that a random low-degree polynomial has large bias. Bhandari et al. \cite{BHSS2022} show that evaluating all low-degree monomials on a sufficiently small random set of points gives linearly independent columns with high probability. We state each result at the point where it is used.

\paragraph{A subspace of a Reed--Muller code}

For $f:\FF_2^m\to\FF_2$, define
\[
  \bias(f)=\E_{x\in\FF_2^m}(-1)^{f(x)}.
\]
Recall that, for integers $0\leq r\leq m$, the Reed--Muller space $\RM(m,r)$ is the $\FF_2$-vector space of functions $f:\FF_2^m\to\FF_2$ represented by multilinear polynomials of total degree at most $r$. It has dimension $\binom{m}{\leq r}$. We use the following bound from \cite{BHL2012} on the probability that a random polynomial has large bias.

\begin{lemma}\cite[Lemma 1.2]{BHL2012}
\label{lem:bhl}
There are constants $\beta_1,\beta_2>0$ such that, whenever $m$ is sufficiently large and $1\leq r\leq m/2$, a uniform polynomial $f\in\RM(m,r)$ satisfies
\[
  \Pr_f\left\{|\bias(f)|>2^{-\beta_1m/r}\right\}
  \leq 2^{-\beta_2\binom{m}{\leq r}}.
\]
Here ``uniform'' means that the coefficients of all multilinear monomials of degree at most $r$ are independent uniform bits.
\end{lemma}

Fix $r\geq3$ and put $\ell_n=\lceil\log_2n\rceil$. Choose fixed constants $b>\max\{1,3r/\beta_1\}$ and $0<\kappa<\min\{\beta_2/3,1/2\}$. For each $n$, set $m=\lceil b\ell_n\rceil$, $R=\binom{m}{\leq r}$, $K=\lfloor\kappa R\rfloor$, and $\eta=2^{-\beta_1m/r}\leq n^{-3}$. The condition $\kappa<1/2$ ensures that $K<R$, so the $K$-dimensional subspaces used below exist, while $\kappa<\beta_2/3$ is used to make the expected number of bad elements tend to zero. Since $r$ is fixed and $m\to\infty$, the condition $1\leq r\leq m/2$ holds for all sufficiently large $n$. We first choose the subspace.

\begin{lemma}
\label{lem:subspace}
For all sufficiently large $n$, there is a $K$-dimensional subspace $W\leq\RM(m,r)$ such that
\[
  |\bias(f)|\leq\eta
  \qquad\text{for every nonzero }f\in W.
\]
\end{lemma}

\begin{proof}
Call a nonzero polynomial bad if its bias has absolute value greater than $\eta$. By Lemma~\ref{lem:bhl}, there are at most $2^{(1-\beta_2)R}$ bad polynomials. Choose $W$ uniformly among the $K$-dimensional subspaces of $\RM(m,r)$. For every fixed nonzero $f$, we have $\Pr\{f\in W\}=(2^K-1)/(2^R-1)\leq2^{K-R+1}$. Hence the expected number of bad elements in $W$ is at most
\[
  2^{(1-\beta_2)R}2^{K-R+1}=2^{K-\beta_2R+1}=o(1).
\]
Here we used $K\leq\kappa R$ and $\kappa<\beta_2/3$. For all sufficiently large $n$, the expectation is smaller than one, so some $W$ contains no bad polynomial.
\end{proof}

Fix such a subspace and a basis $f_1,\ldots,f_K$. Define the map
\[
  \phi:\FF_2^m\to\FF_2^K,
  \qquad
  \phi(x)=\bigl(f_1(x),\ldots,f_K(x)\bigr).
\]
We call $x\in\FF_2^m$ an evaluation point and $\phi(x)$ its evaluation vector with respect to the chosen basis of $W$. For every nonzero $a\in\FF_2^K$, the function $a\cdot\phi(x)=\sum_{k=1}^K a_kf_k(x)$ is a nonzero element of $W$, so
\[
  \left|\E_x(-1)^{a\cdot\phi(x)}\right|\leq\eta.
\]

\begin{remark}
The full space $\RM(m,r)$ does not have this property: it contains the constant-one polynomial, whose absolute bias is one, and monomials such as $x_1\cdots x_r$, whose absolute bias is $1-2^{1-r}$. Lemma~\ref{lem:subspace} removes such functions while preserving dimension $K=\Theta(m^r)$. Thus the subspace makes the Fourier proof of Lemma~\ref{lem:girth} direct; the later rank argument only uses the fact that the coordinate functions of $\phi$ have degree at most $r$.
\end{remark}

\paragraph{$L$-wise linearly independent evaluation vectors}

Set $L=\left\lfloor\frac{K}{4\ell_n}\right\rfloor$ and $D_\star=\left\lfloor\frac L2\right\rfloor=\Theta((\log n)^{r-1})$. We next choose the evaluation points so that their evaluation vectors have no short linear dependence.

\begin{lemma}
\label{lem:girth}
With the subspace and evaluation map fixed above, for all sufficiently large $n$ there are points $x_1,\ldots,x_n\in\FF_2^m$ such that, writing
\[
  v_i:=\phi(x_i)\in\FF_2^K,
  \qquad i\in[n],
\]
the family $(v_i)_{i\in I}$ is linearly independent for every $I\subseteq[n]$ with $|I|\leq L$.
\end{lemma}

\begin{proof}
Sample $x_1,\ldots,x_n$ independently and uniformly from $\FF_2^m$. For a fixed nonempty set $I\subseteq[n]$ of size $w\leq L$, Fourier inversion and independence give
\begin{align*}
  \Pr\left\{\sum_{i\in I}v_i=0\right\}
  &=2^{-K}\sum_{a\in\FF_2^K}
    \E(-1)^{a\cdot\sum_{i\in I}\phi(x_i)}\\
  &=2^{-K}\sum_{a\in\FF_2^K}
    \left(\E_x(-1)^{a\cdot\phi(x)}\right)^w\\
  &\leq2^{-K}+2^{-K}\sum_{a\neq0}
    \left|\E_x(-1)^{a\cdot\phi(x)}\right|^w\\
  &\leq2^{-K}+\eta^w.
\end{align*}
The last inequality bounds every nonzero Fourier term in absolute value. To rule out all short zero sums at once, a union bound gives
\[
  \Pr\left\{\exists\,\varnothing\neq I\subseteq[n],\ |I|\leq L:\sum_{i\in I}v_i=0\right\}
  \leq L2^{L\ell_n-K}+(1+\eta)^n-1
  \leq L2^{-3K/4}+e^{n\eta}-1=o(1).
\]
Here we used $\binom nw\leq n^w$, $L\ell_n\leq K/4$, and $n\eta\leq n^{-2}$. Thus a choice with no zero sum of size at most $L$ exists. Over $\FF_2$, every nontrivial linear dependence is a zero sum over its nonempty support, so $(v_i)_{i\in I}$ is linearly independent whenever $I\subseteq[n]$ and $|I|\leq L$. Finally, since $r$ is fixed, $R=m^r/r!+O(m^{r-1})$, and hence
\[
  D_\star=\left(\frac{\kappa b^r}{8r!}+o(1)\right)
  (\log_2n)^{r-1}.
\]
\end{proof}

\begin{remark}\label{remark-girth}
With respect to the chosen basis of $W$, the columns $(\phi(x))_{x\in\FF_2^m}$ form a generator matrix for the length-$2^m$ evaluation code of $W$. Choosing $x_1,\ldots,x_n$ amounts to puncturing this code to $n$ coordinates. Keeping all coordinates would leave short dependencies: for every affine subspace $H\subseteq\FF_2^m$ of dimension $r+1$, we have $\sum_{x\in H}\phi(x)=0$, because every polynomial of degree at most $r$ sums to zero over $H$. Lemma~\ref{lem:girth} chooses a puncturing whose column matroid has girth greater than $L$; equivalently, the dual of the punctured code has minimum distance greater than $L$. This is exactly what Proposition~\ref{prop:marginals} uses, since $t\leq D_\star$ edges involve at most $2t\leq L$ endpoint vectors. The same general principle appears in \cite{BlasiokEtAl2026}, where an even cover is a nonempty set of hyperedges in which every vertex occurs an even number of times, equivalently a nonzero $\FF_2$-dependence among their incidence vectors. Expansion rules out small even covers and thereby gives local uniformity.
\end{remark}

Choose $n_0$ sufficiently large that $L\geq2$ and the preceding construction applies for every $n\geq n_0$, and fix one such list of evaluation points and evaluation vectors for each $n\geq n_0$.

\subsection{Construction and low-degree advantage}
In this subsection, we use the evaluation vectors and a random alternating form to define the planted distribution and prove that its statistics of degree at most $D_\star$ agree with the null distribution.

\paragraph{Construction of the planted distributions}

Let $\Alt(K)$ denote the vector space of bilinear maps $B:\FF_2^K\times\FF_2^K\to\FF_2$ satisfying $B(u,u)=0$ for every $u\in\FF_2^K$. Such forms are called alternating. In characteristic two, every alternating form is symmetric, because $0=B(u+v,u+v)=B(u,v)+B(v,u)$. If $e_1,\ldots,e_K$ is the standard basis, a uniform $B\in\Alt(K)$ is obtained by choosing the values $B(e_a,e_b)$ for $1\leq a<b\leq K$ as independent uniform bits; symmetry and the identities $B(e_a,e_a)=0$ determine all remaining values.

Fix $n\geq n_0$. Sample $B$ uniformly from $\Alt(K)$ and independently sample a uniform permutation $\pi$ of $[n]$. Define a random graph $X=(X_e)_{e\in\sE_n}$ by
\[
  X_{\{\pi(i),\pi(j)\}}:=B(v_i,v_j),
  \qquad 1\leq i<j\leq n.
\]
The symmetry of $B$ makes this definition independent of the order of the endpoints, and alternation gives zero diagonal if the edge array is extended to an adjacency matrix. Let $\PP_n$ be the distribution of $X$. For every fixed vertex relabeling $\tau$, the graph $\tau X$ is obtained by replacing $\pi$ with $\tau\circ\pi$. Since this composition is again uniform, $\PP_n$ is invariant under every vertex relabeling.

For the finitely many $n<n_0$, set $\PP_n:=\QQ_n$. In the rest of this section, we work with $n\geq n_0$.

\paragraph{Matrix view}

For an equivalent matrix description, let $M_B\in\FF_2^{K\times K}$ be the matrix of $B$ in the standard basis, so that $B(u,v)=u^\top M_Bv$, and place the evaluation vectors into the columns of $V:=\begin{bmatrix}v_1&\cdots&v_n\end{bmatrix}\in\FF_2^{K\times n}$. The matrix $M_B$ is symmetric with zero diagonal, and its entries above the diagonal are independent uniform bits. Before relabeling, the adjacency matrix is $A'=V^\top M_BV$. If $P_\pi e_i=e_{\pi(i)}$ is the permutation matrix of $\pi$, then the observed adjacency matrix is
\[
  A=P_\pi A'P_\pi^\top
   =P_\pi V^\top M_BVP_\pi^\top.
\]
In particular, $A_{\pi(i),\pi(j)}=B(v_i,v_j)$, in agreement with the definition above.

\paragraph{Zero low-degree advantage} We now show that the edges in every small set are independent and uniform.

\begin{proposition}
\label{prop:marginals}
For every $t\leq D_\star$, every marginal of $\PP_n$ on $t$ distinct edges is uniform on $\FF_2^t$. Thus $\PP_n$ is exactly $D_\star$-wise uniform, equivalently its marginals on at most $D_\star$ edge coordinates agree with $\QQ_n$. The same statement holds for $\Tnoise_\eps \PP_n$ for every $\eps\in[0,1]$.
\end{proposition}

\begin{proof}
Fix $t\leq D_\star$ distinct observed edges and condition on $\pi$. An observed vertex $u$ carries the evaluation vector $v_{\pi^{-1}(u)}$. Since $v_i$ and $v_j$ are linearly independent whenever $i\neq j$, these vectors are pairwise distinct. Thus distinct observed vertices carry distinct evaluation vectors, and the selected distinct edges give distinct unordered pairs of vectors. Let $w_1,\ldots,w_h$, where $h\leq2t\leq L$, be the distinct evaluation vectors carried by the endpoints of the selected edges. They are linearly independent by Lemma~\ref{lem:girth}; put
\[
  U=\operatorname{span}\{w_1,\ldots,w_h\}.
\]

To identify the joint law of the edge values, we show that the restriction map from $\Alt(K)$ to the space $\Alt(U)$ of alternating bilinear forms on $U$ is surjective. Indeed, given any $C\in\Alt(U)$, extend $w_1,\ldots,w_h$ to a basis of $\FF_2^K$, define $B$ to agree with $C$ on $U\times U$, set all pairings involving an added basis vector to zero, and extend bilinearly. Then $B\in\Alt(K)$ and $B|_{U\times U}=C$.

Since restriction is a surjective linear map between finite vector spaces, every element of $\Alt(U)$ has the same number of extensions in $\Alt(K)$; hence a uniform $B\in\Alt(K)$ restricts to a uniform element of $\Alt(U)$. In the basis $w_1,\ldots,w_h$, the values
\[
  B(w_p,w_q),\qquad 1\leq p<q\leq h,
\]
are therefore independent uniform bits. The edge indicators under consideration are therefore independent uniform bits; averaging over $\pi$ preserves this conclusion.

It remains to check resampling. Condition on the set of edges that are replaced. Among the selected coordinates, the unreplaced ones form a marginal of the uniform vector just obtained, while the replaced ones are independent uniform bits.
\end{proof}

This uniformity gives the low-degree conclusion.

\begin{corollary}
\label{cor:adv}
Let $\sL_n=\frac{\mathrm d\PP_n}{\mathrm d\QQ_n}$. For every integer $0\leq d\leq D_\star$,
\[
  \Adv_{\leq d}(\PP_n,\QQ_n)=0.
\]
Therefore, $\sL_n^{\leq d}=1$ and $\|\sL_n^{\leq d}\|_{L^2(\QQ_n)}=1$.
\end{corollary}

\section{Noise-robust polynomial-time distinguisher}
\label{sec:distinguisher}

In this section, we define the rank test and prove that it distinguishes $\Tnoise_\eps\PP_n$ from $\QQ_n$ in polynomial time.

\subsection{The rank test}

Fix rational constants $a,\sigma>0$ as in Lemma~\ref{lem:parameters}. Set $T=\lfloor a\ell_n\rfloor$, $s=\lfloor\sigma\ell_n\rfloor$, and $N_0=\lfloor\sqrt n\rfloor$. For the finitely many $n$ with $T+N_0>n$, define the test to return ``null.'' Henceforth assume $T+N_0\leq n$. For a graph $Y\in\{0,1\}^{\sE_n}$, write $Y_{ij}:=Y_{\{i,j\}}$. Use vertices $1,\ldots,T$ as anchors. For every selected nonanchor $j\in\{T+1,\ldots,T+N_0\}$, form the vector of its edges to the anchors,
\[
  y_j=(Y_{1j},\ldots,Y_{Tj})\in\FF_2^T.
\]
Define the monomial lift
\[
  \Psi_s(y)=\left(\prod_{i\in S}y_i\right)_{S\subseteq[T],\ |S|\leq s}
  \in\FF_2^{R_0},
  \qquad
  R_0=\binom{T}{\leq s}.
\]
The subsets are placed in a fixed order, and the empty product is one. Let $M_s(Y)$ be the $R_0\times N_0$ matrix with columns $\Psi_s(y_j)$. Its column-rank deficit is $N_0-\rank M_s(Y)$.

The test declares ``planted'' if and only if the lifted columns are linearly dependent, that is,
\[
  \rank M_s(Y)<N_0.
\]
The test does not use the resampling rate.

Algorithm~\ref{alg:rank-test} gives pseudocode for the rank test; all matrix operations are over $\FF_2$.

\begin{algorithm}[htbp]
\DontPrintSemicolon
\LinesNumbered
\caption{Rank test}\label{alg:rank-test}
\KwIn{An observed graph $Y$ on vertex set $[n]$.}
\KwOut{A label in $\{\textsc{Planted},\textsc{Null}\}$.}

Set $T\gets\lfloor a\ell_n\rfloor$ and $N_0\gets\lfloor\sqrt n\rfloor$\;

\If{$T+N_0>n$}{
    \Return \textsc{Null}\;
}

Set $s\gets\lfloor\sigma\ell_n\rfloor$\;

\For{$j=T+1,T+2,\dots,T+N_0$}{
    Set $y_j\gets(Y_{1j},\ldots,Y_{Tj})\in\FF_2^T$\;

    \ForEach{$S\subseteq[T]$ with $|S|\leq s$}{
        Set $(\Psi_s(y_j))_S\gets\prod_{i\in S}(y_j)_i$\;
    }
}

Set $M_s(Y)\gets[\Psi_s(y_{T+1})\ \cdots\ \Psi_s(y_{T+N_0})]$\;

Compute $\rho_Y\gets\rank M_s(Y)$ by Gaussian elimination\;

\eIf{$\rho_Y<N_0$}{
    \Return \textsc{Planted}\;
}{
    \Return \textsc{Null}\;
}
\end{algorithm}

\subsection{Full rank under the null distribution}

The analysis under $\QQ_n$ uses the following result on random puncturings of Reed--Muller codes.

\begin{lemma}\cite[Theorem 1.1 and Corollary 4.1]{BHSS2022}
\label{lem:bhss}
There is an absolute constant $\gamma_0>0$ with the following property. Let $T\to\infty$, let $s=s(T)\to\infty$ be integers satisfying $s<\gamma_0T$, and fix $\delta_0\in(0,1)$. If $Z$ is a uniformly random subset of $\FF_2^T$ of size $\lfloor(1-\delta_0)\binom{T}{\leq s}\rfloor$, then the vectors $\{\Psi_s(z):z\in Z\}$ are linearly independent with probability $1-o(1)$.
\end{lemma}

\begin{remark}
The random puncturing in Lemma~\ref{lem:bhss} has a different role from the puncturing used to define $\PP_n$. In Section~\ref{sec:distributions}, the evaluation points are fixed once and give column girth greater than $L$. Under $\QQ_n$, the random vectors $y_j$ instead select random columns from the generator matrix of $\RM(T,s)$, and Lemma~\ref{lem:bhss} gives full column rank.
\end{remark}

\begin{lemma}
\label{lem:null}
Suppose $a>2$, $\sigma/a<\min\{\gamma_0,1/2\}$, and $a h_2(\sigma/a)>1/2$. If $Y\sim \QQ_n$, then
\[
  \rank M_s(Y)=N_0
\]
with probability $1-o(1)$.
\end{lemma}

\begin{proof}
Under $\QQ_n$, the vectors $y_j$ are independent and uniform in $\FF_2^T$. By a union bound, the probability that two of them are equal is at most $\binom{N_0}{2}2^{-T}=n^{1-a+o(1)}=o(1)$. Conditional on all the vectors being distinct, the set $U:=\{y_{T+1},\ldots,y_{T+N_0}\}$ is a uniformly random $N_0$-subset of $\FF_2^T$.

Since $\sigma/a<1/2$,
\[
  R_0=\binom{T}{\leq s}
     =n^{a h_2(\sigma/a)+o(1)}\gg N_0.
\]
Thus $N_0\leq N_1:=\lfloor R_0/2\rfloor$ for all sufficiently large $n$.

The puncturing theorem applies to a random set whose size is a fixed fraction of the code dimension. We therefore enlarge the sampled set to this size. Given $U$, choose $Z$ uniformly among the $N_1$-subsets of $\FF_2^T$ containing $U$. Every fixed $N_1$-subset contains exactly $\binom{N_1}{N_0}$ possible sets $U$, and every $U$ has the same number of $N_1$-supersets. Hence the marginal distribution of $Z$ is uniform among all $N_1$-subsets of $\FF_2^T$. Since $T,s\to\infty$ and $s/T=\sigma/a+o(1)<\gamma_0$, Lemma~\ref{lem:bhss} with $\delta_0=1/2$ shows that the columns indexed by $Z$ are linearly independent with probability $1-o(1)$. The original columns form a subset of this family and are therefore also independent. Since all the vectors $y_j$ are distinct with probability $1-o(1)$, the result follows.
\end{proof}

\subsection{Rank under the noisy planted distribution}

We next bound the matrix rank under $\Tnoise_\eps\PP_n$. For a fixed $\eps\in[0,1)$, define
\[
  \lambda=-\log_2(1-\eps),
  \qquad
  \mu_n=N_0(1-\eps)^T=n^{1/2-a\lambda+o(1)}.
\]

\begin{lemma}
\label{lem:planted}
Suppose $r\sigma/b<1/2$ and $b h_2(r\sigma/b)<1/2-a\lambda$. If $Y\sim\Tnoise_\eps \PP_n$, then
\[
  N_0-\rank M_s(Y)\geq\mu_n/2
\]
with probability $1-o(1)$.
\end{lemma}

\begin{proof}
Condition on $\pi$ and $B$, and write $z_j:=x_{\pi^{-1}(j)}$ for every observed vertex $j$. For each anchor $i\in[T]$, define $g_i(x):=B(\phi(z_i),\phi(x))$. Since $u\mapsto B(\phi(z_i),u)$ is a linear functional on $\FF_2^K$, there is a vector $c_i\in\FF_2^K$ such that $B(\phi(z_i),u)=c_i\cdot u$ for every $u\in\FF_2^K$. Consequently,
\[
  g_i(x)=c_i\cdot\phi(x)=\sum_{k=1}^K(c_i)_kf_k(x)\in W\leq\RM(m,r).
\]

Next we count the columns whose anchor edges survive resampling. Call a selected nonanchor column clean if none of its $T$ anchor edges was replaced, and let $N_{\mathrm{clean}}$ be the number of clean columns. The relevant replacement decisions are disjoint for different nonanchors, so $N_{\mathrm{clean}}\sim\Bin(N_0,(1-\eps)^T)$ with mean $\mu_n$. The assumption $b h_2(r\sigma/b)<1/2-a\lambda$ implies $1/2-a\lambda>0$, and hence $\mu_n\to\infty$. A Chernoff bound therefore gives $N_{\mathrm{clean}}\geq3\mu_n/4$ with probability $1-o(1)$.

We now bound the rank contributed by these clean columns. If a nonanchor $j$ is clean, then
\[
  y_j=\bigl(g_1(z_j),\ldots,g_T(z_j)\bigr).
\]
For each $S\subseteq[T]$ with $|S|\leq s$, define
\[
  p_S(x):=\prod_{i\in S}g_i(x).
\]
After multilinear reduction on $\FF_2^m$, each $p_S$ lies in $\RM(m,rs)$. Let $M_{\mathrm{clean}}$ be the submatrix formed by the clean columns. Its row indexed by $S$ consists of the values $p_S(z_j)$ over the clean columns, so every row lies in the image of the evaluation map from $\RM(m,rs)$ to $\FF_2^{N_{\mathrm{clean}}}$. Since $rs/m=r\sigma/b+o(1)<1/2$,
\[
  \rank M_{\mathrm{clean}}
  \leq\dim\RM(m,rs)
  =\binom{m}{\leq rs}
  =n^{b h_2(r\sigma/b)+o(1)}
  =o(\mu_n).
\]
Finally, each of the $N_0-N_{\mathrm{clean}}$ remaining columns adds at most one to the rank, giving
\[
  \rank M_s(Y)\leq N_0-N_{\mathrm{clean}}+\binom{m}{\leq rs}.
\]
For all sufficiently large $n$, the Reed--Muller dimension $\binom{m}{\leq rs}$ is at most $\mu_n/4$. Together with $N_{\mathrm{clean}}\geq3\mu_n/4$, this gives $N_0-\rank M_s(Y)\geq\mu_n/2$ with probability $1-o(1)$. Since $\mu_n\to\infty$, the lifted columns are linearly dependent with probability $1-o(1)$.
\end{proof}

\subsection{Completion of the proof}

It remains to choose constants that meet the assumptions used in both rank bounds.

\begin{lemma}
\label{lem:parameters}
For every fixed integer $r\geq3$ and constant $b>0$, there are rational constants $\sigma,a>0$ and a constant $\eps_0\in(0,1)$ such that the hypotheses of Lemmas~\ref{lem:null} and \ref{lem:planted} hold for every $\eps\in[0,\eps_0]$, with $a h_2(\sigma/a)<1$.
\end{lemma}

\begin{proof}
Choose rational $\sigma>0$ so small that
\[
  r\sigma/b<1/2,\quad b h_2(r\sigma/b)<1/4,\quad
  \sigma/2<\min\{\gamma_0,1/2\},\quad 2h_2(\sigma/2)<1/2.
\]
The function $a\mapsto a h_2(\sigma/a)$ is continuous for $a\geq2$ and tends to infinity as $a\to\infty$. We can therefore choose rational $a>2$ with
\[
  1/2<a h_2(\sigma/a)<1.
\]
We also have $\sigma/a<\min\{\gamma_0,1/2\}$. Finally, choose $\eps_0\in(0,1)$ so small that $a[-\log_2(1-\eps_0)]<1/4$. For every $\eps\in[0,\eps_0]$, we have
\[
  a\lambda\leq a[-\log_2(1-\eps_0)]<1/4.
\]
Hence $b h_2(r\sigma/b)<1/4<1/2-a\lambda$, so all assumptions of Lemmas~\ref{lem:null} and~\ref{lem:planted} hold.
\end{proof}

We now combine the two rank bounds to prove strong distinction. To finish the main theorem, we also check the running time, the low-degree agreement, and the sampling claim.

\begin{proof}[Proof of Theorem~\ref{thm:main}]
Choose the constants from Lemma~\ref{lem:parameters} and set $\eps_r:=\eps_0$. Under $\QQ_n$, Lemma~\ref{lem:null} gives $\rank M_s(Y)=N_0$ with probability $1-o(1)$, so the test outputs ``null.'' For every fixed $\eps\in[0,\eps_r]$, Lemma~\ref{lem:planted} gives a positive rank deficit with probability $1-o(1)$ under $\Tnoise_\eps\PP_n$, so the test outputs ``planted.'' Thus the same test strongly distinguishes at every such rate, without knowing $\eps$.

Next we check the running time. Let $\alpha:=a h_2(\sigma/a)\in(1/2,1)$. Then $R_0=n^{\alpha+o(1)}$ and $N_0=\lfloor\sqrt n\rfloor$. Constructing $M_s(Y)$ takes $O(N_0R_0T)=n^{1/2+\alpha+o(1)}$ field operations, while ordinary Gaussian elimination takes $O(R_0N_0^2)=n^{1+\alpha+o(1)}=o(n^2)$. These are operations over $\FF_2$.

For the low-degree claim, choose $c_r$ with $0<c_r<\kappa b^r/(8r!)$. Then $D_n\leq D_\star$ for all sufficiently large $n$, so Proposition~\ref{prop:marginals} and Corollary~\ref{cor:adv} prove item~1.

For the sampling claim, first consider $n<n_0$: the sampler draws directly from $\PP_n=\QQ_n$. For $n\geq n_0$, store the points $x_1,\ldots,x_n$ and the coefficients of $f_1,\ldots,f_K$ in the multilinear monomial basis of $\RM(m,r)$. The points and coefficients use $nm+KR=O(n\log n+(\log n)^{2r})$ advice bits, with $O(\log n)$ further bits to record the dimensions. They permit reconstruction of every $v_i=\phi(x_i)$ in polynomial time. A randomized algorithm then samples the $\binom K2$ independent values $B(e_a,e_b)$ for $a<b$, equivalently the entries of $M_B$ above the diagonal, and samples a uniform permutation $\pi$ of $[n]$. It evaluates every edge and therefore samples exactly from the same $\PP_n$.
\end{proof}

\section{Discussion and future work}
\label{sec:discussion}

\paragraph{Explicitness.}
The sampler in Theorem~\ref{thm:main} uses advice because our proof shows only that a suitable subspace $W$ and suitable evaluation points exist. Both can be made more explicit. With $m=12\ell_n$ and $t=3\ell_n$, consider the trace functions $x\mapsto\operatorname{Tr}(c\,x^{1+2^{i_2}+\cdots+2^{i_r}})$ on $\FF_{2^m}\cong\FF_2^m$, where $c$ ranges over a basis of $\FF_{2^m}$ and $1\leq i_2<\cdots<i_r<t$. They have degree at most $r$, and by the Weil bound \cite[Lemma~1]{MakhulWinterhof2022} every nonzero linear combination has absolute bias at most $2^{t-m/2}\leq n^{-3}$; hence they span a subspace of dimension $\Theta((\log n)^r)$ that can be constructed in polynomial time. Choosing $L$-wise independent evaluation points requires only a seed of $O((\log n)^r)$ bits. Since every event in the proof of Lemma~\ref{lem:girth} involves at most $L$ points, some seed is good, and Theorem~\ref{thm:main} then holds with $b=12$ and possibly different constants. This seed is the only advice the sampler needs; a good seed can also be found by brute force in $\exp(O_r((\log n)^r))$ time, and the distinguisher uses neither. Finding a good seed in polynomial time, and thus a uniformly polynomial-time samplable example, remains open.

\paragraph{Sublinear-time distinguisher.}
Our rank test also gives a sublinear-time distinguisher in the random-access model. The same planted family can be tested using fewer columns, with a smaller allowed noise rate if necessary. For example, using $\lfloor n^{1/4}\rfloor$ columns gives $o(n)$ time and $O(n^{1/4}\log n)$ queries. A full input scan still takes $O(n^2)$ time.

We verify this by adjusting the constants in the two rank bounds. Fix a rational $\theta\in(0,1/2]$ and replace $N_0$ by $\lfloor n^\theta\rfloor$. Choose rational $\sigma>0$ so small that
\[
  \begin{gathered}
    r\sigma/b<1/2,\qquad b h_2(r\sigma/b)<\theta/2,\\
    \sigma/2<\min\{\gamma_0,1/2\},\qquad 2h_2(\sigma/2)<\theta.
  \end{gathered}
\]
As in Lemma~\ref{lem:parameters}, choose rational $a>2$ with $\theta<\alpha:=a h_2(\sigma/a)<2\theta$, and then choose $\eps_{r,\theta}\in(0,1)$ with $a[-\log_2(1-\eps_{r,\theta})]<\theta/2$.

Under the null, the collision probability is at most $n^{2\theta-a+o(1)}=o(1)$ and $R_0=n^{\alpha+o(1)}\gg N_0$, so the proof of Lemma~\ref{lem:null} still gives full column rank.

At every fixed rate $\eps\in[0,\eps_{r,\theta}]$, the mean number of clean columns is $n^{\theta-a\lambda+o(1)}$, whose exponent exceeds $\theta/2$. The clean-column rank bound has exponent $b h_2(r\sigma/b)<\theta/2$, so the proof of Lemma~\ref{lem:planted} still gives a positive rank deficit.

The test uses $O(n^\theta\log n)$ edge queries and $n^{\alpha+2\theta+o(1)}=o(n^{4\theta})$ random-access time.

\paragraph{Higher-order observations.}
This counterexample extends readily from graphs to symmetric tensors of any fixed order $k\geq2$. Using the same evaluation vectors, replace $B$ by a uniform alternating $k$-linear form and set $X_{\{i_1,\ldots,i_k\}}=B(v_{i_1},\ldots,v_{i_k})$, followed by uniform vertex relabeling; alternating forms are symmetric over $\FF_2$. Since $t$ hyperedges involve at most $kt$ vectors, the proof of Proposition~\ref{prop:marginals} gives exact low-degree agreement through degree $\lfloor L/k\rfloor=\Theta((\log n)^{r-1})$. For the test, replace the anchors by $T$ disjoint $(k-1)$-sets avoiding the $N_0$ column vertices. Fixing these arguments still leaves a polynomial of degree at most $r$ in the remaining latent point, so $T$, $N_0$, $s$, and the noise threshold $\eps_r$ need no change. The same analysis gives $O(\sqrt n\log n)$ queries and $o(n^2)$ random-access time; a full input scan now takes $O(n^k)$ time.

\paragraph{Gaussian observations.}
The half-normal embedding from the proof of \cite[Lemma~6]{KlivansKothari2014} replaces each bit $X_e$ by $(-1)^{X_e}|G_e|$, where the $G_e$ are independent standard Gaussians independent of $X$. The null becomes product standard Gaussian, and Ornstein--Uhlenbeck noise takes the place of coordinate resampling. Vertex-relabeling invariance and agreement on small marginals are preserved, including after noise, so real polynomials have zero advantage through the same degree. For sufficiently weak fixed noise, thresholding at zero yields the Boolean resampling model \cite[Remark~2.2]{BHJK2025}, so the same rank test applies with the same query and running-time bounds, treating coordinate access and sign comparison as unit-cost operations.

\paragraph{Algebraic structure and possible repairs.}
Our counterexample shows how low-dimensional algebraic structure can defeat low-degree predictions: Gaussian elimination detects a rank deficit in a monomial lift of a submatrix, even after sufficiently small fixed-rate resampling. A related phenomenon appears in code-based cryptography: Schur product attacks exploit unusually small product spaces in certain Reed--Solomon variants of the McEliece cryptosystem \cite{McEliece1978,CGGOT2014}. This suggests that a repaired low-degree conjecture should exclude such structure when it remains efficiently detectable after resampling. Finding natural assumptions that do so remains an open problem.

\section*{Statement on the use of AI}
\addcontentsline{toc}{section}{Statement on the use of AI}

The author used the ChatGPT family for assistance with the proofs. In the initial prompts, the author presented several constructions using tensor codes based on Reed--Muller codes and other codes, including constructions similar in structure to the present one, which uses a random codeword from a tensor subcode. The author asked whether these constructions could be made invariant under all vertex relabelings while preserving an efficient decoding algorithm. GPT 5.4 responded that the proposed approach could not satisfy both requirements.

Inspired by~\cite{JiaVijayaraghavan2026}, the author then conjectured that the subspace structure of these codes could produce a separation in small eigenvalues or singular values. The author applied random permutations to codewords from the same codes and examined eigenvalues and singular values after applying different matrix transformations. At this stage, the constructions lacked the subspace selection and puncturing steps, making it difficult to obtain both zero low-degree advantage and sufficient separation in the tested statistics. The author then fed the ideas in Remark~\ref{remark-girth} to GPT 5.6 and asked whether these ideas could be used to reduce the low-degree advantage while preserving separation. Several rounds of discussion led to the final construction and rank argument.

Building on the previous version, the author used GPT 6 to improve the distinguisher's running time from polynomial time to linear time in the graph input size.

The author independently checked, simplified, and organized every proof and takes full responsibility for all statements and results in the paper.

\section*{Acknowledgments}
\addcontentsline{toc}{section}{Acknowledgments}

This work was completed during the author's wonderful visit to Bocconi University in Milan in the summer of 2026. The author thanks Alon Rosen for his warm hospitality during the visit and for suggesting the use of small-bias spaces for derandomization. The author is especially grateful to Tim Kunisky for first introducing the author to the low-degree conjecture in 2024 and for early discussions on this topic and later discussions of related problems, which provided much of the intuition behind this project. Tim Kunisky also provided many helpful comments and suggestions on the presentation of an earlier version.

\bibliographystyle{alpha}
\begingroup
\raggedright
\phantomsection
\addcontentsline{toc}{section}{\refname}
\bibliography{references}

@article{BHL2012,
  author = {Ido Ben-Eliezer and Rani Hod and Shachar Lovett},
  title = {Random Low-Degree Polynomials Are Hard to Approximate},
  journal = {Computational Complexity},
  volume = {21},
  number = {1},
  pages = {63--81},
  year = {2012},
  doi = {10.1007/s00037-011-0020-6},
  url = {https://doi.org/10.1007/s00037-011-0020-6}
}

@inproceedings{BHSS2022,
  author = {Siddharth Bhandari and Prahladh Harsha and Ramprasad Saptharishi and Srikanth Srinivasan},
  title = {Vanishing Spaces of Random Sets and Applications to {Reed--Muller} Codes},
  booktitle = {37th Computational Complexity Conference (CCC 2022)},
  series = {Leibniz International Proceedings in Informatics (LIPIcs)},
  volume = {234},
  pages = {31:1--31:14},
  year = {2022},
  doi = {10.4230/LIPIcs.CCC.2022.31},
  url = {https://doi.org/10.4230/LIPIcs.CCC.2022.31}
}

@inproceedings{BHJK2025,
  author = {Rares-Darius Buhai and Jun-Ting Hsieh and Aayush Jain and Pravesh K. Kothari},
  title = {The Quasi-Polynomial Low-Degree Conjecture Is False},
  booktitle = {2025 {IEEE} 66th Annual Symposium on Foundations of Computer Science ({FOCS})},
  pages = {2577--2590},
  publisher = {IEEE Computer Society},
  year = {2025},
  doi = {10.1109/FOCS63196.2025.00134},
  url = {https://doi.org/10.1109/FOCS63196.2025.00134}
}

@inproceedings{KlivansKothari2014,
  author = {Adam Klivans and Pravesh Kothari},
  title = {Embedding Hard Learning Problems Into {Gaussian} Space},
  booktitle = {Approximation, Randomization, and Combinatorial Optimization. Algorithms and Techniques (APPROX/RANDOM 2014)},
  series = {Leibniz International Proceedings in Informatics (LIPIcs)},
  volume = {28},
  pages = {793--809},
  year = {2014},
  doi = {10.4230/LIPIcs.APPROX-RANDOM.2014.793},
  url = {https://doi.org/10.4230/LIPIcs.APPROX-RANDOM.2014.793}
}

@article{McEliece1978,
  author = {Robert J. McEliece},
  title = {A Public-Key Cryptosystem Based on Algebraic Coding Theory},
  journal = {DSN Progress Report},
  volume = {42-44},
  pages = {114--116},
  year = {1978},
  url = {https://ipnpr.jpl.nasa.gov/progress_report2/42-44/44N.PDF}
}

@article{CGGOT2014,
  author = {Alain Couvreur and Philippe Gaborit and Val{\'e}rie Gauthier-Uma{\~n}a and Ayoub Otmani and Jean-Pierre Tillich},
  title = {Distinguisher-Based Attacks on Public-Key Cryptosystems Using {Reed--Solomon} Codes},
  journal = {Designs, Codes and Cryptography},
  volume = {73},
  number = {2},
  pages = {641--666},
  year = {2014},
  doi = {10.1007/s10623-014-9967-z},
  url = {https://doi.org/10.1007/s10623-014-9967-z}
}

@article{MakhulWinterhof2022,
  author = {Mehdi Makhul and Arne Winterhof},
  title = {Normality of the {Thue--Morse} Function for Finite Fields along Polynomial Values},
  journal = {Research in Number Theory},
  volume = {8},
  pages = {38},
  year = {2022},
  doi = {10.1007/s40993-022-00335-8},
  url = {https://doi.org/10.1007/s40993-022-00335-8}
}

@phdthesis{Hopkins2018,
  author = {Samuel B. Hopkins},
  title = {Statistical Inference and the Sum of Squares Method},
  school = {Cornell University},
  year = {2018},
  doi = {10.7298/X4416V82},
  url = {https://hdl.handle.net/1813/59618}
}

@inproceedings{HW2021,
  author = {Justin Holmgren and Alexander S. Wein},
  title = {Counterexamples to the Low-Degree Conjecture},
  booktitle = {12th Innovations in Theoretical Computer Science Conference (ITCS 2021)},
  series = {Leibniz International Proceedings in Informatics (LIPIcs)},
  volume = {185},
  pages = {75:1--75:9},
  year = {2021},
  doi = {10.4230/LIPIcs.ITCS.2021.75},
  url = {https://doi.org/10.4230/LIPIcs.ITCS.2021.75}
}

@incollection{KWB2019,
  author = {Dmitriy Kunisky and Alexander S. Wein and Afonso S. Bandeira},
  title = {Notes on Computational Hardness of Hypothesis Testing: Predictions Using the Low-Degree Likelihood Ratio},
  booktitle = {Mathematical Analysis, Its Applications and Computation},
  series = {Springer Proceedings in Mathematics \& Statistics},
  volume = {385},
  pages = {1--50},
  publisher = {Springer},
  year = {2022},
  doi = {10.1007/978-3-030-97127-4_1},
  url = {https://doi.org/10.1007/978-3-030-97127-4_1}
}

@misc{Wein2025Survey,
  author = {Alexander S. Wein},
  title = {Computational Complexity of Statistics: New Insights from Low-Degree Polynomials},
  year = {2025},
  eprint = {2506.10748},
  archiveprefix = {arXiv},
  primaryclass = {math.ST},
  doi = {10.48550/arXiv.2506.10748},
  note = {arXiv:2506.10748},
  url = {https://arxiv.org/abs/2506.10748}
}

@article{BHKMP2019,
  author = {Boaz Barak and Samuel B. Hopkins and Jonathan A. Kelner and Pravesh K. Kothari and Ankur Moitra and Aaron Potechin},
  title = {A Nearly Tight Sum-of-Squares Lower Bound for the Planted Clique Problem},
  journal = {SIAM Journal on Computing},
  volume = {48},
  number = {2},
  pages = {687--735},
  year = {2019},
  doi = {10.1137/17M1138236},
  url = {https://doi.org/10.1137/17M1138236}
}

@inproceedings{HKPSS2017,
  author = {Samuel B. Hopkins and Pravesh K. Kothari and Aaron Potechin and Prasad Raghavendra and Tselil Schramm and David Steurer},
  title = {The Power of Sum-of-Squares for Detecting Hidden Structures},
  booktitle = {2017 IEEE 58th Annual Symposium on Foundations of Computer Science (FOCS)},
  pages = {720--731},
  year = {2017},
  doi = {10.1109/FOCS.2017.72},
  url = {https://doi.org/10.1109/FOCS.2017.72}
}

@inproceedings{HopkinsSteurer2017,
  author = {Samuel B. Hopkins and David Steurer},
  title = {Efficient {Bayesian} Estimation from Few Samples: Community Detection and Related Problems},
  booktitle = {2017 IEEE 58th Annual Symposium on Foundations of Computer Science (FOCS)},
  pages = {379--390},
  year = {2017},
  doi = {10.1109/FOCS.2017.42},
  url = {https://doi.org/10.1109/FOCS.2017.42}
}

@inproceedings{BBKMW2021,
  author = {Afonso S. Bandeira and Jess Banks and Dmitriy Kunisky and Cristopher Moore and Alexander S. Wein},
  title = {Spectral Planting and the Hardness of Refuting Cuts, Colorability, and Communities in Random Graphs},
  booktitle = {Proceedings of the Thirty-Fourth Conference on Learning Theory},
  series = {Proceedings of Machine Learning Research},
  volume = {134},
  pages = {410--473},
  year = {2021},
  url = {https://proceedings.mlr.press/v134/bandeira21a.html}
}

@inproceedings{BandeiraKuniskyWein2020,
  author = {Afonso S. Bandeira and Dmitriy Kunisky and Alexander S. Wein},
  title = {Computational Hardness of Certifying Bounds on Constrained {PCA} Problems},
  booktitle = {11th Innovations in Theoretical Computer Science Conference (ITCS 2020)},
  series = {Leibniz International Proceedings in Informatics (LIPIcs)},
  volume = {151},
  pages = {78:1--78:29},
  year = {2020},
  doi = {10.4230/LIPIcs.ITCS.2020.78},
  url = {https://doi.org/10.4230/LIPIcs.ITCS.2020.78}
}

@article{DKWB2021,
  author = {Yunzi Ding and Dmitriy Kunisky and Alexander S. Wein and Afonso S. Bandeira},
  title = {The Average-Case Time Complexity of Certifying the Restricted Isometry Property},
  journal = {IEEE Transactions on Information Theory},
  volume = {67},
  number = {11},
  pages = {7355--7361},
  year = {2021},
  doi = {10.1109/TIT.2021.3112823},
  url = {https://doi.org/10.1109/TIT.2021.3112823}
}

@article{DKWB2024,
  author = {Yunzi Ding and Dmitriy Kunisky and Alexander S. Wein and Afonso S. Bandeira},
  title = {Subexponential-Time Algorithms for Sparse {PCA}},
  journal = {Foundations of Computational Mathematics},
  volume = {24},
  number = {3},
  pages = {865--914},
  year = {2024},
  doi = {10.1007/s10208-023-09603-0},
  url = {https://doi.org/10.1007/s10208-023-09603-0}
}

@article{DingDuLi2025,
  author = {Jian Ding and Hang Du and Zhangsong Li},
  title = {Low-Degree Hardness of Detection for Correlated {Erd{\H{o}}s--R{\'e}nyi} Graphs},
  journal = {The Annals of Statistics},
  volume = {53},
  number = {5},
  pages = {1833--1856},
  year = {2025},
  doi = {10.1214/25-AOS2517},
  url = {https://doi.org/10.1214/25-AOS2517}
}

@article{SchrammWein2022,
  author = {Tselil Schramm and Alexander S. Wein},
  title = {Computational Barriers to Estimation from Low-Degree Polynomials},
  journal = {The Annals of Statistics},
  volume = {50},
  number = {3},
  pages = {1833--1858},
  year = {2022},
  doi = {10.1214/22-AOS2179},
  url = {https://doi.org/10.1214/22-AOS2179}
}

@misc{Mao2026,
  author = {Songtao Mao},
  title = {Near Optimal Algorithms for Noisy {$k$}-{XOR} under Low-Degree Heuristic},
  year = {2026},
  eprint = {2604.10457},
  archiveprefix = {arXiv},
  primaryclass = {cs.CC},
  doi = {10.48550/arXiv.2604.10457},
  note = {arXiv:2604.10457},
  url = {https://arxiv.org/abs/2604.10457}
}

@inproceedings{BBHLS2021,
  author = {Matthew S. Brennan and Guy Bresler and Samuel B. Hopkins and Jerry Li and Tselil Schramm},
  title = {Statistical Query Algorithms and Low Degree Tests Are Almost Equivalent},
  booktitle = {Proceedings of the Thirty-Fourth Conference on Learning Theory},
  series = {Proceedings of Machine Learning Research},
  volume = {134},
  pages = {774--774},
  year = {2021},
  url = {https://proceedings.mlr.press/v134/brennan21a.html}
}

@article{MontanariWein2025,
  author = {Andrea Montanari and Alexander S. Wein},
  title = {Equivalence of Approximate Message Passing and Low-Degree Polynomials in Rank-One Matrix Estimation},
  journal = {Probability Theory and Related Fields},
  volume = {191},
  pages = {181--233},
  year = {2025},
  doi = {10.1007/s00440-024-01322-z},
  url = {https://doi.org/10.1007/s00440-024-01322-z}
}

@inproceedings{ChooDorsi2021,
  author = {Davin Choo and Tommaso d'Orsi},
  title = {The Complexity of Sparse Tensor {PCA}},
  booktitle = {Advances in Neural Information Processing Systems},
  volume = {34},
  pages = {7993--8005},
  year = {2021},
  url = {https://proceedings.neurips.cc/paper/2021/hash/42a6845a557bef704ad8ac9cb4461d43-Abstract.html}
}

@inproceedings{BreslerHuang2022,
  author = {Guy Bresler and Brice Huang},
  title = {The Algorithmic Phase Transition of Random {$k$}-{SAT} for Low Degree Polynomials},
  booktitle = {2021 IEEE 62nd Annual Symposium on Foundations of Computer Science (FOCS)},
  pages = {298--309},
  year = {2022},
  eprint = {2106.02129},
  archiveprefix = {arXiv},
  primaryclass = {cs.DS},
  url = {https://arxiv.org/abs/2106.02129}
}

@article{GamarnikJagannathWein2024,
  author = {David Gamarnik and Aukosh Jagannath and Alexander S. Wein},
  title = {Hardness of Random Optimization Problems for {Boolean} Circuits, Low-Degree Polynomials, and {Langevin} Dynamics},
  journal = {SIAM Journal on Computing},
  volume = {53},
  number = {1},
  pages = {1--46},
  year = {2024},
  doi = {10.1137/22M150263X},
  url = {https://doi.org/10.1137/22M150263X}
}

@inproceedings{BogdanovKothariRosen2023,
  author = {Andrej Bogdanov and Pravesh K. Kothari and Alon Rosen},
  title = {Public-Key Encryption, Local Pseudorandom Generators, and the Low-Degree Method},
  booktitle = {Theory of Cryptography},
  series = {Lecture Notes in Computer Science},
  volume = {14369},
  pages = {268--285},
  publisher = {Springer},
  year = {2023},
  doi = {10.1007/978-3-031-48615-9_10},
  url = {https://doi.org/10.1007/978-3-031-48615-9_10}
}

@inproceedings{JiaVijayaraghavan2026,
  author = {He Jia and Aravindan Vijayaraghavan},
  title = {Low-Degree Method Fails to Predict Robust Subspace Recovery},
  booktitle = {Proceedings of the Thirty-Ninth Conference on Learning Theory},
  series = {Proceedings of Machine Learning Research},
  volume = {336},
  pages = {3751--3781},
  publisher = {PMLR},
  year = {2026},
  url = {https://proceedings.mlr.press/v336/jia26a.html}
}

@misc{BlasiokEtAl2026,
  author = {Jaros{\l}aw B{\l}asiok and Paul Lou and Alon Rosen and Madhu Sudan},
  title = {Expanders Meet {Reed--Muller}: Easy Instances of Noisy {$k$-{XOR}}},
  year = {2026},
  eprint = {2604.04188},
  archiveprefix = {arXiv},
  primaryclass = {cs.CC},
  doi = {10.48550/arXiv.2604.04188},
  note = {arXiv:2604.04188v1},
  url = {https://arxiv.org/abs/2604.04188v1}
}

@article{ChenDingGongLi2026,
  author = {Guanyi Chen and Jian Ding and Shuyang Gong and Zhangsong Li},
  title = {A Computational Transition for Detecting Correlated Stochastic Block Models by Low-Degree Polynomials},
  journal = {The Annals of Statistics},
  volume = {54},
  number = {1},
  pages = {226--251},
  year = {2026},
  doi = {10.1214/25-AOS2565},
  url = {https://doi.org/10.1214/25-AOS2565}
}

@article{DhawanMaoWein2025,
  author = {Abhishek Dhawan and Cheng Mao and Alexander S. Wein},
  title = {Detection of Dense Subhypergraphs by Low-Degree Polynomials},
  journal = {Random Structures \& Algorithms},
  volume = {66},
  number = {1},
  pages = {e21279},
  year = {2025},
  doi = {10.1002/rsa.21279},
  url = {https://doi.org/10.1002/rsa.21279}
}

@inproceedings{KuniskyYu2024,
  author = {Dmitriy Kunisky and Xifan Yu},
  title = {Computational Hardness of Detecting Graph Lifts and Certifying Lift-Monotone Properties of Random Regular Graphs},
  booktitle = {2024 IEEE 65th Annual Symposium on Foundations of Computer Science (FOCS)},
  pages = {1621--1633},
  year = {2024},
  doi = {10.1109/FOCS61266.2024.00101},
  url = {https://arxiv.org/abs/2404.17012}
}

@article{MaoWein2025,
  author = {Cheng Mao and Alexander S. Wein},
  title = {Optimal Spectral Recovery of a Planted Vector in a Subspace},
  journal = {Bernoulli},
  volume = {31},
  number = {2},
  pages = {1114--1139},
  year = {2025},
  doi = {10.3150/24-BEJ1763},
  url = {https://doi.org/10.3150/24-BEJ1763}
}

@misc{Li2025CorrelatedSpiked,
  author = {Zhangsong Li},
  title = {The Algorithmic Phase Transition in Correlated Spiked Models},
  year = {2025},
  eprint = {2511.06040},
  archiveprefix = {arXiv},
  primaryclass = {math.ST},
  doi = {10.48550/arXiv.2511.06040},
  note = {arXiv:2511.06040},
  url = {https://arxiv.org/abs/2511.06040}
}

@inproceedings{BandeiraEtAl2022FranzParisi,
  author = {Afonso S. Bandeira and Ahmed El Alaoui and Samuel B. Hopkins and Tselil Schramm and Alexander S. Wein and Ilias Zadik},
  title = {The {Franz--Parisi} Criterion and Computational Trade-offs in High Dimensional Statistics},
  booktitle = {Advances in Neural Information Processing Systems},
  volume = {35},
  pages = {33831--33844},
  year = {2022},
  url = {https://proceedings.neurips.cc/paper_files/paper/2022/hash/daff682411a64632e083b9d6665b1d30-Abstract-Conference.html}
}

@article{Wein2021MIS,
  author = {Alexander S. Wein},
  title = {Optimal Low-Degree Hardness of Maximum Independent Set},
  journal = {Mathematical Statistics and Learning},
  volume = {4},
  number = {3/4},
  pages = {221--251},
  year = {2021},
  doi = {10.4171/MSL/25},
  url = {https://doi.org/10.4171/MSL/25}
}

@inproceedings{hsieh2026rigorous,
  author = {Jun-Ting Hsieh and Daniel M. Kane and Pravesh K. Kothari and Jerry Li and Sidhanth Mohanty and Stefan Tiegel},
  title = {Rigorous Implications of the {Low-Degree} Heuristic},
  booktitle = {Proceedings of the 58th Annual ACM Symposium on Theory of Computing},
  pages = {1763--1770},
  publisher = {Association for Computing Machinery},
  year = {2026},
  doi = {10.1145/3798129.3800883},
  url = {https://doi.org/10.1145/3798129.3800883}
}

@inproceedings{DK2022,
  author = {Ilias Diakonikolas and Daniel Kane},
  title = {Non-Gaussian Component Analysis via Lattice Basis Reduction},
  booktitle = {Proceedings of Thirty Fifth Conference on Learning Theory},
  series = {Proceedings of Machine Learning Research},
  volume = {178},
  pages = {4535--4547},
  year = {2022},
  publisher = {PMLR},
  url = {https://proceedings.mlr.press/v178/diakonikolas22d.html}
}

@inproceedings{ZSWB2022,
  author = {Ilias Zadik and Min Jae Song and Alexander S. Wein and Joan Bruna},
  title = {Lattice-Based Methods Surpass Sum-of-Squares in Clustering},
  booktitle = {Proceedings of Thirty Fifth Conference on Learning Theory},
  series = {Proceedings of Machine Learning Research},
  volume = {178},
  pages = {1247--1248},
  year = {2022},
  publisher = {PMLR},
  eprint = {2112.03898},
  archiveprefix = {arXiv},
  primaryclass = {cs.LG},
  url = {https://proceedings.mlr.press/v178/zadik22a.html}
}

@inproceedings{BangachevBresler2024,
  author = {Kiril Bangachev and Guy Bresler},
  title = {On the {Fourier} Coefficients of High-Dimensional Random Geometric Graphs},
  booktitle = {Proceedings of the 56th Annual ACM Symposium on Theory of Computing},
  pages = {549--560},
  publisher = {ACM},
  year = {2024},
  doi = {10.1145/3618260.3649676},
  url = {https://doi.org/10.1145/3618260.3649676}
}

@article{LuoGao2024,
  author = {Yuetian Luo and Chao Gao},
  title = {Computational Lower Bounds for Graphon Estimation via Low-Degree Polynomials},
  journal = {The Annals of Statistics},
  volume = {52},
  number = {5},
  pages = {2318--2348},
  year = {2024},
  doi = {10.1214/24-AOS2437},
  url = {https://doi.org/10.1214/24-AOS2437}
}

@inproceedings{Wein2023TensorDecomposition,
  author = {Alexander S. Wein},
  title = {Average-Case Complexity of Tensor Decomposition for Low-Degree Polynomials},
  booktitle = {Proceedings of the 55th Annual ACM Symposium on Theory of Computing},
  pages = {1685--1698},
  publisher = {ACM},
  year = {2023},
  doi = {10.1145/3564246.3585232},
  url = {https://doi.org/10.1145/3564246.3585232}
}

@article{MoitraWein2025,
  author = {Ankur Moitra and Alexander S. Wein},
  title = {Precise Error Rates for Computationally Efficient Testing},
  journal = {The Annals of Statistics},
  volume = {53},
  number = {2},
  pages = {854--878},
  year = {2025},
  doi = {10.1214/25-AOS2490},
  url = {https://doi.org/10.1214/25-AOS2490}
}

@inproceedings{ArpinoVenkataramanan2023,
  author = {Gabriel Arpino and Ramji Venkataramanan},
  title = {Statistical-Computational Tradeoffs in Mixed Sparse Linear Regression},
  booktitle = {Proceedings of Thirty Sixth Conference on Learning Theory},
  series = {Proceedings of Machine Learning Research},
  volume = {195},
  pages = {921--986},
  publisher = {PMLR},
  year = {2023},
  url = {https://proceedings.mlr.press/v195/arpino23a.html}
}

@inproceedings{DingHuaSlotSteurer2025,
  author = {Jingqiu Ding and Yiding Hua and Lucas Slot and David Steurer},
  title = {Low-Degree Evidence for Computational Transition of Recovery Rate in Stochastic Block Model},
  booktitle = {Advances in Neural Information Processing Systems},
  volume = {38},
  year = {2025},
  doi = {10.52202/085713-3846},
  url = {https://papers.nips.cc/paper_files/paper/2025/hash/a789f19f6454f858ad74bc8b7581fe36-Abstract-Conference.html}
}

@misc{Li2025Contiguity,
  author = {Zhangsong Li},
  title = {Computational Lower Bounds for Correlated Random Graphs via Algorithmic Contiguity},
  year = {2025},
  eprint = {2502.09832},
  archiveprefix = {arXiv},
  primaryclass = {stat.ML},
  doi = {10.48550/arXiv.2502.09832},
  note = {arXiv:2502.09832v4},
  url = {https://arxiv.org/abs/2502.09832v4}
}
\endgroup

\end{document}